\documentclass[a4paper,UKenglish,cleveref, autoref, thm-restate]{lipics-v2021}

\DeclareMathOperator{\polylog}{polylog}

\def\point{\tt{point}}
\def\rankop{\tt{rank}}

\title{Space Efficient Two-Dimensional Orthogonal Colored Range Counting} 

\titlerunning{Space-Efficient Orthogonal Colored Range Counting} 

\author{Younan Gao}{Faculty of Computer Science, Dalhousie University, Canada }{yn803382@dal.ca}{}{}

\author{Meng He}{Faculty of Computer Science, Dalhousie University, Canada}{mhe@cs.dal.ca}{}{}

\authorrunning{Y. Gao and M. He} 

\Copyright{Younan Gao and Meng He} 



\category{} 

\relatedversion{} 

\funding{This work was supported by NSERC of Canada}

\nolinenumbers 

\hideLIPIcs  

\EventEditors{Petra Mutzel, Rasmus Pagh, and Grzegorz Herman}
\EventNoEds{3}
\EventLongTitle{29th Annual European Symposium on Algorithms (ESA 2021)}
\EventShortTitle{ESA 2021}
\EventAcronym{ESA}
\EventYear{2021}
\EventDate{September 6--8, 2021}
\EventLocation{Lisbon, Portugal}
\EventLogo{}
\SeriesVolume{204}
\ArticleNo{73}

\begin{document}
	
	\maketitle
	
	\begin{abstract}
		 In the two-dimensional orthogonal colored range counting problem, we preprocess a set, $P$, of $n$ colored points on the plane, such that given an orthogonal query rectangle, the number of distinct colors of the points contained in this rectangle can be computed efficiently.
	For this problem, we design three new solutions, and the bounds of each can be expressed in some form of time-space tradeoff.
		By setting appropriate parameter values for these solutions, we can achieve new specific results with (the space are in words and $\epsilon$ is an arbitrary constant in $(0,1)$):
		\begin{itemize}
			\item $O(n\lg^3 n)$ space and $O(\sqrt{n}\lg^{5/2} n \lg \lg n)$ query time;
			\item $O(n\lg^2 n)$ space and $O(\sqrt{n}\lg^{4+\epsilon} n)$ query time;
			\item $O(n\frac{\lg^2 n}{\lg \lg n})$ space and $O(\sqrt{n}\lg^{5+\epsilon} n)$ query time;
			\item $O(n\lg n)$ space and $O(n^{1/2+\epsilon})$ query time.
		\end{itemize}
		A known conditional lower bound to this problem based on Boolean matrix multiplication gives some evidence on the difficulty of achieving near-linear space solutions with query time better than $\sqrt{n}$ by more than a polylogarithmic factor using purely combinatorial approaches. Thus the time and space bounds in all these results are efficient. 
		Previously, among solutions with similar query times, the most space-efficient solution uses $O(n\lg^4 n)$ space to answer queries in $O(\sqrt{n}\lg^8 n)$ time (SIAM. J. Comp.~2008).
		Thus the new results listed above all achieve improvements in space efficiency, while all but the last result achieve speed-up in query time as well. 
	\end{abstract}
	
	\newpage
	
	\section{Introduction}
	\label{sec:introduction}
	
		In computational geometry, there have been extensive studies on problems over points associated with information represented as colors \cite{gupta1995further,gupta1996algorithms, kaplan2006colored, kaplan2008efficient, lai2008approximate, grossi2014colored, nekrich2014efficient, el2017succinct, gupta2018computational, chan2020better, chan2020further, pathcoloredcounting2021, rahul2021approximate}. 
	Among them, the {\em 2D orthogonal range counting} query problem is one of the most fundamental.
	In this problem, we preprocess a set, $P$, of $n$ points on the plane, each colored in one of $C$ different colors, such that given an orthogonal query rectangle, the number of distinct colors of the points contained in this rectangle can be computed efficiently.
	
	This problem is important in both theory and practice. Theoretically, it has connections to matrix multiplication: The ability to answer $m$ colored range counting queries offline over $n$ points on the plane in $o(\min\{n,m\}^{\omega/2})$ time, where $\omega$ is the best current exponent of the running time of matrix multiplication, would yield a faster algorithm for Boolean matrix multiplication~\cite{kaplan2008efficient}. 
	In practice, the records in database systems and many other applications are often associated with categorical information which can be modeled as colors. 
	The Structured Programming Language (SQL) thus provides keywords such as {\texttt{DISTINCT}} and {\texttt{GROUP BY}} to compute information about the distinct categories of the records within a query range, which can be modeled using colored range query problems, 
	and these queries have also been used in database query optimization~\cite{ccmn2000}.
	
	One challenge in solving the 2D orthogonal range counting problem is that the queries are not easily decomposible: if we partition the query range into two or more subranges, we cannot simply obtain the number of distinct colors in the query range by adding up the number of distinct colors in each subrange.
	Furthermore, the conditional lower bound based on matrix multiplication as described above gives theoretical evidence on the hardness of this problem.
	Indeed, if polylogarithmic query times are desired, the solution with the best space efficiency~\cite{munro2015range} uses $O(n^2\lg n/\lg \lg n)$ words of space to answer queries in $O((\lg n/\lg \lg n)^2)$ time.
	There is a big gap between the complexities of this solution to those of the optimal solution to {\em 2D orthogonal range counting} for which we do not have color information and are only interested in computing the number of points in a 2D orthogonal query range.
	The latter can be solved in merely linear space and $O(\lg n/\lg \lg n)$ query time~\cite{jaja2004space}.

	Applications that process a significant amount of data would typically require structures whose space costs are much lower than quadratic. 
	As the running time of the best known combinatorial algorithm of multiplying two $n\times{}n$ Boolean matrices is $\Theta(n^3/\polylog(n))$~\cite{bw2012,c2015,y2018},
	the conditional lower bound of 2D orthogonal colored range counting implies that no solution can simultaneously have preprocessing time better than $\Omega(n^{3/2})$ and query time better than $\Omega(\sqrt{n})$, by purely combinatorial methods\footnote{When algebraic approaches are allowed, $\omega \approx 2.3727$~\cite{williams2012multiplying}, implying that the preprocessing time and the query time cannot be simultaneously less than $\Omega(n^{1.18635})$ and $\Omega(n^{0.18635})$, respectively.} with current knowledge, save for polylogarithmic speed-ups.
	To match this query time within polylogarithmic factors, the most space-efficient solution uses $O(n\lg^4 n)$ words of space to answer queries in $O(\sqrt{n}\lg^8 n)$ time~\cite{kaplan2008efficient}.
	Despite this breakthrough, the exponents in the polylogarithmic factors in the time and space bounds leave much room for potential improvements. 
	Hence, in this paper, we aim at decreasing these polylogarithmic factors in both space and time costs, to design solutions that are more desirable for applications that manage large data sets. 
	
	\subparagraph*{Previous Work.} 
	Gupta et al.~\cite{gupta1995further} showed how to reduce the orthogonal colored range searching problem in 1D to orthogonal range searching over uncolored points in 2D, thus achieving a linear-space solution with $O(\lg n/\lg \lg n)$ query time.
	Later, the query time was improved to $O(\lg C/\lg \lg n)$ by Nekrich~\cite{nekrich2014efficient}, where $C$ is the number of colors.
	
	To solve 2D colored orthogonal range counting, Gupta et al.~\cite{gupta1995further} used persistent data structures to extend their 1D solution to 2D and designed a data structure of $O(n^2\lg^2 n)$ words that supports queries in $O(\lg^2 n)$ time. 
	Kaplan et al.~\cite{kaplan2008efficient} achieved the same bounds by decomposing the input points into disjoint boxes in 3D and reduced the problem to 3D stabbing counting queries.
	Recently, Munro at al. \cite{munro2015range} showed that 3D 3-sided colored range counting can be answered in $O((\lg n/\lg \lg n)^2)$ time using a data structure of $O(n(\lg n/\lg \lg n))$ words of space, which implies a solution to 2D 3-sided colored range counting with the same time and space bound.
	For each distinct $x$-coordinate $x_i$ of the points in the point set, if we use the strategy in \cite{gupta1995further, kaplan2008efficient} to build a data structure supporting 2D 3-sided queries upon the points whose $x$-coordinates are greater than or equal to $x_i$, then this set of data structures constructed can be used to answer a 4-sided query.
	This yields a solution to the 2D orthogonal colored range counting problem with 
	$O(n^2\lg n/\lg \lg n)$ words of space and $O((\lg n/\lg \lg n)^2)$ query time.
	Kaplan et al.~\cite{kaplan2008efficient} further showed how to achieve time-space tradeoffs by designing a solution with $O(X\lg^7 n)$ query time that uses $O((\frac{n}{X})^2\lg^6 n+n\lg^4 n)$ words of space.
	Setting $X = \sqrt{n}\lg n$ minimizes space usage, achieving an $O(n\lg^4 n)$-word solution with $O(\sqrt{n}\lg^8 n)$ query time.
	When only linear space is allowed, Grossi and Vind \cite{grossi2014colored} showed how to answer a query in $O(n/\polylog(n))$ time.
	Though not explicitly stated anywhere, by combining an approach that Kaplan et al.~\cite{kaplan2008efficient} presented for dimensions of $3$ or higher (which also works for 2D) and a linear space solution to 2D orthogonal range emptiness~\cite{chan2011orthogonal}, the query time can be improved to $O(n^{3/4}\lg^{\epsilon} n)$ for any constant $\epsilon$ using a linear space structure.
	Finally, Kaplan et al. also considered the offline version of this problem and showed that $n$ 2D orthogonal colored range counting queries can be answered in $O(n^{1.408})$ time.
	
	Researchers have also studied approximate colored range counting problem. 
	In 1D, El-Zein et al.~\cite{el2017succinct} designed a succinct data structure to answer a $(1+\epsilon)$-approximate colored counting query in constant time.
	In 2D, Rahul~\cite{rahul2021approximate} provided a reduction from $(1+\epsilon)$-approximate orthogonal colored range counting to 2D colored orthogonal range reporting which reports the number of distinct colors in a 2D orthogonal query range. Based on this, they gave an $O(n\lg n)$-word data structure with $O(\lg n)$ query time.

	The orthogonal colored range counting problem has also been studied in higher dimensions~\cite{lai2008approximate,kaplan2008efficient, rahul2021approximate, pathcoloredcounting2021}. 
	Furthermore, He and Kazi~\cite{pathcoloredcounting2021} generalized it to categorical path counting by replacing the first dimension with tree topology.
	One of their solutions to a path query problem generalized from 2D orthogonal colored range counting also uses linear space and provides $O(n^{3/4}\lg^{\epsilon} n)$ query time. 
	We end this brief survey by commenting that, after a long serious of work \cite{janardan1993generalized, gupta1995further, kaplan2008efficient, nekrich2014efficient, grossi2014colored, gupta2018computational}, 
	Chan and Nekrich~\cite{chan2020better} solved the related 2D orthogonal range reporting problem in $O(n\lg^{3/4+\epsilon} n)$ words of space and $O(\lg \lg n+k)$ query time for points in rank space, where $k$ is output size.
	This almost matches the bounds of the optimal solution to (uncolored) 2D orthogonal range reporting over points in rank space, which uses $O(n\lg^{\epsilon} n)$ words to answer queries in $O(\lg \lg n+k)$ time.
	
	\begin{table}[t]
		\centering
		\caption{\label{tab} Bounds of 2D orthogonal colored range counting structures. 
		  The results in the form of time-space tradeoffs are listed in the top portion, in which $X$ and $\lambda$ are integer parameters in $[1,n]$ and $[2, n]$, respectively.
                  The bottom portion presents results with specific bounds, among which marked with a $^{\dagger}$ are those obtained from the top portion by setting appropriate parameters values. }
		\begin{threeparttable}
			\begin{tabularx}{\columnwidth}{l|l|l|X}
				\hline
				Source & Model & Query Time & Space Usage in Words \\
				\hline
				\cite{kaplan2008efficient}& PM & $O(X\lg^7 n)$ & $O((\frac{n}{X})^2\lg^6 n+n\lg^4 n)$ \\
				\hline
				Cor.~\ref{theorem: color_counting_pm} &PM& $O(\lg^5 n+X\lg^{3} n)$ & $O((\frac{n}{X})^2 \lg^4 n+ n\lg^3 n)$ \\
				\hline
				Thm.~\ref{theorem: color_counting_0} &RAM& $O(\lg^4 n+X\lg^{2} n \lg \lg n)$ & $O((\frac{n}{X})^2 \lg^4 n+ n\lg^3 n)$ \\
				\hline
				Thm.~\ref{theorem: color_counting_1} &RAM& $O(\lg^6 n + X\lg^{3+\epsilon} n)$ & $O((\frac{n}{X})^2 \lg^4 n+ n\lg^2 n)$ \\
				\hline
                                Thm.~\ref{theorem:ugly_tradeoff} &RAM& $O(\lambda^2 \lg^6 n \log^2_{\lambda} n + X \lg^{3+\epsilon} n \lambda \log_{\lambda} n)$ & $O((\frac{n}{X})^2 \lg^2 n \log^2_{\lambda} n+n\lg n \log_{\lambda} n)$ \\
				\hline \hline
				\cite{gupta1995further,kaplan2008efficient} & PM & $O(\lg^2 n)$ & $O(n^2\lg^2 n)$ \\
				\hline
				\cite{munro2015range} & RAM & $O((\lg n/\lg \lg n)^2)$ & $O(n^2\lg n/\lg \lg n)$ \\
				\hline
				\cite{kaplan2008efficient}$^{\dagger}$ & PM & $O(\sqrt{n}\lg^8 n)$ & $O(n\lg^4 n)$   \\
				\hline
			        Cor.~\ref{theorem: color_counting_pm}$^{\dagger}$ & PM & $O(\sqrt{n}\lg^{7/2} n)$ & $O(n\lg^3 n)$   \\
				\hline
				Thm.~\ref{theorem: color_counting_0}$^{\dagger}$ & RAM & $O(\sqrt{n}\lg^{5/2} n \lg \lg n)$ & $O(n\lg^3 n)$   \\
				\hline
				Thm.~\ref{theorem: color_counting_1}$^{\dagger}$ & RAM & $O(\sqrt{n}\lg^{4+\epsilon} n)$ & $O(n\lg^2 n)$  \\
				\hline
				Thm.~\ref{theorem:ugly_tradeoff}$^{\dagger}$ &RAM& $O(\sqrt{n}\lg^{5+\epsilon} n)$ & $O(n\frac{\lg^2 n}{\lg \lg n})$ \\
				\hline
				Thm.~\ref{theorem:ugly_tradeoff}$^{\dagger}$ &RAM& $O(n^{1/2+\epsilon})$ & $O(n\lg n)$ \\
				\hline
				\cite{kaplan2008efficient} &PM& $O(n^{3/4}\lg n)$ & $O(n\lg n)$ \\
				\hline
				\cite{grossi2014colored}& RAM & $O(n/\polylog(n))$ & $O(n)$ \\
				\hline
				\cite{kaplan2008efficient,chan2011orthogonal} & RAM & $O(n^{3/4}\lg^{\epsilon} n)$ & $O(n)$ \\
				\hline
			\end{tabularx}
		\end{threeparttable}
	\end{table}
	%
	\subparagraph*{Our Results.} 
	Under the word RAM model, we present three results, all in the form of time-space tradeoffs, for two-dimensional orthogonal colored range counting.
	Specifically, for an integer parameter $X\in[1, n]$, we propose solutions (all space costs are in words): 
	\begin{itemize}
		\item  with $O((\frac{n}{X})^2 \lg^4 n+ n\lg^3 n)$ space and $O(\lg^4 n+X\lg^{2} n \lg \lg n)$ query time; setting $X = \sqrt{n\lg n}$ achieves $O(n\lg^3 n)$ space and $O(\sqrt{n}\lg^{5/2} \lg \lg n)$ query time;
		\item  with $O((\frac{n}{X})^2 \lg^4 n+ n\lg^2 n)$ space and $O(\lg^6 n + X\lg^{3+\epsilon} n)$ query time for any constant $\epsilon\in(0,1)$; setting $X = \sqrt{n}\lg n$ achieves $O(n\lg^2 n)$ space and $O(\sqrt{n}\lg^{4+\epsilon})$ query time;
		\item  with $O((\frac{n}{X})^2 \lg^2 n \cdot \log^2_{\lambda} n+n\lg n \cdot \log_{\lambda} n)$ space and $O(\lambda^2\cdot \lg^6 n \cdot \log^2_{\lambda} n + X \cdot \lg^{3+\epsilon} n \cdot\lambda \log_{\lambda} n)$ query time for an integer parameter $\lambda\in[2, n]$; setting $X = \sqrt{n \lg n \log_{\lambda} n}$ and $\lambda = \lg^{\epsilon} n$ achieves $O(n\frac{\lg^2 n}{\lg \lg n})$ space and $O(\sqrt{n}\lg^{5+\epsilon'} n)$ query time for any $\epsilon'>2\epsilon$, while setting $X=\sqrt{n \lg n}$ and $\lambda = n^{\epsilon/5}$ achieves $O(n\lg n)$ space and $O(n^{1/2+\epsilon})$ query time.
	\end{itemize}
        
	When presenting each result, we also showed the bounds of the most space-efficient tradeoff that can be achieved by setting appropriate parameter values.
        The conditional lower bound based on Boolean matrix multiplication which we discussed before gives some evidence on the difficulty of achieving query time better than $\sqrt{n}$ by more than a polylogarithmic factor using combinatorial approaches without increasing these space costs polynomially.

        When comparing to previous results, note that only the time-space tradeoff presented by Kaplan et al.~\cite{kaplan2008efficient} could possibly achieve near-linear space and $O(\sqrt{n}\polylog(n))$ query time.
        More specifically, their solution uses $O((\frac{n}{X})^2\lg^6 n+n\lg^4 n)$ words of space to achieve $O(X\lg^7 n)$ query time.
        The most space-efficient tradeoff that could be obtained from it is an $O(n\lg^4 n)$-word structure with $O(\sqrt{n}\lg^8 n)$ query time.
        Thus we indeed achieve the goal of improving the polylogarithmic terms in both their time and space costs significantly.

        It is worthwhile to mention that the result of Kaplan et al. can work under the pointer machine (PM) model. Thus, for an absolutely fair comparison, we show how our first result can be adapted to the same model of computation to achieve $O((\frac{n}{X})^2 \lg^4 n+ n\lg^3 n)$ space and $O(\lg^5 n+X\lg^{3} n)$ query time. Thus under PM, we have an $O(n\lg^3 n)$-word structure with $O(\sqrt{n}\lg^{7/2} \lg n)$ query time.
        This is still a significant improvement over previous similar results.
        In the rest of the paper, however, we assume the word RAM model of computation unless otherwise specified, since most of our results are designed under it. 
        See Table~\ref{tab} for a comparison of our results to all previous results. 

        To achieve these results, we use the standard technique of decomposing a 4-sided query range to two 3-sided subranges with a range tree.
        Then the answer can be obtained by adding up the numbers of distinct colors assigned to points in each subrange and then subtracting the number of distinct colors that exist in both.
        We still use an approach of Kaplan et al. to reduce 2D 3-sided colored range counting to 3D stabbing queries over a set of boxes.
        What is new is our scheme of achieving time-space tradeoffs when computing the number of colors that exist in both subranges.
        Based on a parameter, we selectively precompute the sizes of the intersections between pairs of colors sets, each of which corresponds to a prefix of a certain box list somewhere in the stabbing query structures.
        Compared to the scheme of Kaplan et al. for the same purpose, ours gives more flexibility in the design of the 3D stabbing query structures that could work with the scheme.
        This extra flexibility further allows us to use and design different stabbing query structures to achieve new results. 
	\section{Preliminaries}
	\label{sec:prelim}
	In this section, we introduce some notation and previous results used in our paper.
	
	\subparagraph*{Notation.}
	Throughout this paper, we assume points are in general positions unless otherwise specified. 
	In three-dimensional space, we call a box $B$ {\em canonical} if it is defined in the form of $[x_1, +\infty)\times[y_1, y_2)\times [z_1, z_2)$, where $x_1, y_1, z_1 \in {\rm I\!R}$ and $y_2, z_2 \in {\rm I\!R} \cup \{+\infty\}$. 
	We use $B.x_1$ to refer to the lower bound of the $x$-range of $B$, 
	$B.y_1$ and $B.y_2$ to respectively refer to the lower and upper bounds of the $y$-range of $B$, and so on.
	Let $(p.x, p.y, p.z)$ denote the coordinates of a point $p$. 
	We say a point $q\in {\rm I\!R}^3$ {\em dominates} another point $p\in {\rm I\!R}^3$, if $ q.x \ge  p.x$, $q.y \ge p.y$ and $q.z \ge p.z$ hold simultaneously. 

	\subparagraph*{2D Orthogonal Colored Range Emptiness} An {\em orthogonal range emptiness} query determines whether an axis-aligned query rectangle contains at least one point in the point set $P$.
	Observe that a solution to this query problem directly leads to a solution to the colored version of this problem called {\em orthogonal colored range emptiness},
	in which each point in $P$ is colored in one of $C$ different colors, and given a color $c$ and an axis-aligned rectangle, the query asks whether the query range contains at least one point colored in $c$.
	The reduction works as follows: 
	For each color $1 \leq c\leq C$, let $P_c$ denote the subset of $P$ containing all points colored in $c$.
	If we construct an orthogonal range emptiness structure over $P_c$ for each color $c$, then we can answer an orthogonal colored range emptiness query by querying the structure constructed over the points with the query color. 
	The following lemma thus directly follows from the work of Chan et al.~\cite{chan2011orthogonal} on range emptiness:
	
\begin{lemma}[\cite{chan2011orthogonal}]
	\label{lemma_range_emptiness}
	Given $n$ colored points in 2-dimensional rank space, there is a data structure of $nf(n)$ words that answers 2D orthogonal colored range emptiness queries in $g(n)$ time, where 
	\begin{itemize}
		\item a) $f(n)=O(1)$ and $g(n)=O(\lg^{\epsilon} n)$ for any constant $\epsilon>0$; or
		\item b) $f(n)=O(\lg \lg n)$ and $g(n)=O(\lg \lg n)$.
	\end{itemize}
\end{lemma}
	
	\subparagraph*{Orthogonal Stabbing Queries over 3D Canonical Boxes.}
	In the {\em 3D stabbing counting}  problem, we preprocess a set of 3D boxes, such that, given a query point $q$, we can compute the number of boxes containing $q$ efficiently, while in the {\em 3D stabbing reporting} query problem, we report these boxes. 
	In both our solution and the solution of Kaplan et al.~\cite{kaplan2008efficient}, we use data structures for this problem in a special case in which each box is a canonical box.

	Here we introduce the data structure of Kaplan et al.~\cite{kaplan2008efficient} that solves the stabbing query problems over a set of $n$ canonical boxes in 3D, as our first solution to the colored range counting problem augments this data structure and uses it as a component.
	Their data structure consists of two layers of segment trees.
	The structure at the top layer is a segment tree constructed over the $z$-coordinates of the boxes.
	More precisely, we project each box onto the $z$-axis to obtain an interval, and the segment tree is constructed over all these intervals.
	A box is assigned to a node in this tree if its corresponding interval on the $z$-axis is {\em associated} with this node. 
	(Recall that as a segment tree, the leaves correspond to the elementary intervals induced by the endpoints of the intervals projected by each box onto the $z$-axis; and each internal node $v$ of the tree corresponds to an interval $I_z(v)$ that are the union of elementary intervals of the leaves in the subtree rooted at $v$.
	A box is assigned to a node $v$ in this tree if its corresponding interval on the $z$-axis covers the interval $I_z(v)$ but does not cover the interval $I_z(u)$, where $u$ is the parent node of $v$.)
	For each node $v$ in the top-layer segment tree, we further construct a segment tree over the projections of the boxes assigned to $v$ on the $y$-axis, and the segment trees constructed for all the nodes of the top-layer tree form the bottom-layer structure. 
	Finally, for each node, $v'$, of a segment tree at the bottom layer, we construct a list, $sList(v')$, of the boxes assigned to $v'$, sorted by their $x_1$-coordinates, i.e. the left endpoints of their projections on the $x$-axis; recall that the right endpoint is always $+\infty$ by the definition of canonical boxes. 
	The lists of boxes associated with the nodes in the same bottom-level segment tree are linked together in a fractional cascading \cite{chazelle1986fractional} data structure to facilitate the search among $x_1$-coordinates of the boxes in each list. 
	
	
	With these data structures, the query algorithm first performs a search in the top-layer segment tree to locate the $O(\lg n)$ nodes whose associated intervals (which are projections of boxes on the $z$-axis) contain the $z$-coordinate of the query point $q$.
	For each of these nodes, we then query the bottom-layer segment tree constructed for it to find the boxes whose projections on the $yz$-plane contain $(q.y, q.z)$. 
	The searches with this layer of structures end at $O(\lg^2 n)$ nodes of the bottom-layer trees, whose $sList$'s are disjoint.
	Finally, a binary search at the $sList$ of each of these nodes, sped up by factional cascading, gives us the answer. Thus:

	%
	\begin{lemma}[\cite{kaplan2008efficient}]
		\label{lemma:stabbing_4}
		Given a set of $n$ canonical boxes in three dimension, the above data structure occupies $O(n\lg^2 n)$ words and answers stabbing counting queries in $O(\lg^2 n)$ time and stabbing reporting queries in $O(\lg^2 n+k)$ time, where $k$ denotes the number of boxes reported. 
		Furthermore, the output of the reporting query is the union of $O(\lg^2 n)$ disjoint subsets, each containing the boxes stored in a nonempty prefix of a  bottom-layer sorted list. 
		The preprocessing time is $O(n\lg^2 n)$.
	\end{lemma}

	\subparagraph*{Reducing 2D 3-Sided Colored Range Counting to 3D Orthogonal Stabbing Counting over Canonical Boxes.}
	\label{sect: color_counting_reduction}
A key technique used in both the solutions of  Kaplan et al.~\cite{kaplan2008efficient} and our solutions is a reduction from 2D 3-sided colored range counting queries, in which each query range is of the form  $[a, b] \times [c, +\infty)$ for some $a, b, c \in {\rm I\!R}$, to 3D orthogonal stabbing queries over canonical boxes\footnote{Later we sometimes deal with query ranges of the form $[a, b] \times (\infty,d]$ for some $a, b, d \in {\rm I\!R}$, and a similar reduction also works.}.
	The reduction is performed in two steps.
	First we reduce the 2D 3-sided colored range counting query problem to the 3D dominance colored range counting problem, in which we preprocess a set, $P$, of colored points in  ${\rm I\!R}^3$, such that given a query point $q$ in ${\rm I\!R}^3$, one can report the number of distinct colors in $P\cap (-\infty, q.x]\times(-\infty, q.y]\times(-\infty, q.z]$ efficiently.
	This reduction works as follows: For each point, $p=(p.x, p.y)$, we create a point, $p'=(-p.x, p.x, -p.y)$, in  ${\rm I\!R}^3$, and assign it with the color of $p$.
	Then a 2D 3-sided colored range counting query over the original points in which the query range is $[a, b] \times [c, +\infty)$ can be answered by performing a 3D dominance colored range query over the created points, using $(-\infty, -a]\times(-\infty, b]\times(-\infty, -c]$ as the query range. 
	
	To further reduce the 3D dominance colored range counting problem to 3D orthogonal stabbing counting over canonical boxes, we need some additional notation:
	Given a point $p$ in 3D, let $Q^{+}_p$ denote region $[p.x, +\infty)\times[p.y, +\infty)\times[p.z, +\infty)$.
	Furthermore, given a point set $A$, let $U(A)$ denote the region of $\cup_{p\in A} Q^{+}_p$.
	Then the following lemma is crucial:
	
	\begin{lemma}[\cite{kaplan2008efficient}]
		\label{lemma_box}
		Given a set, $A$, of $n$ points in three-dimensional space, a set of $O(n)$ pairwise disjoint 3D canonical boxes can be computed in $O(n\lg^2 n)$ time such that the union of these boxes is the region $U(A)$. 
	\end{lemma}
	
	Originally Kaplan et al.~\cite[Theorem 2.1]{kaplan2008efficient} proved the above lemma for $d$-dimensional space where $d \ge 1$; for a general $d$, the number of boxes required to cover $U(A)$ is $O(n^{\lfloor d/2\rfloor})$. To help readers understand this decomposition, Figure~\ref{fig-box} shows an example for $d=2$.
	
	\begin{figure}[h]
		\centering
		{\includegraphics[scale=0.5]{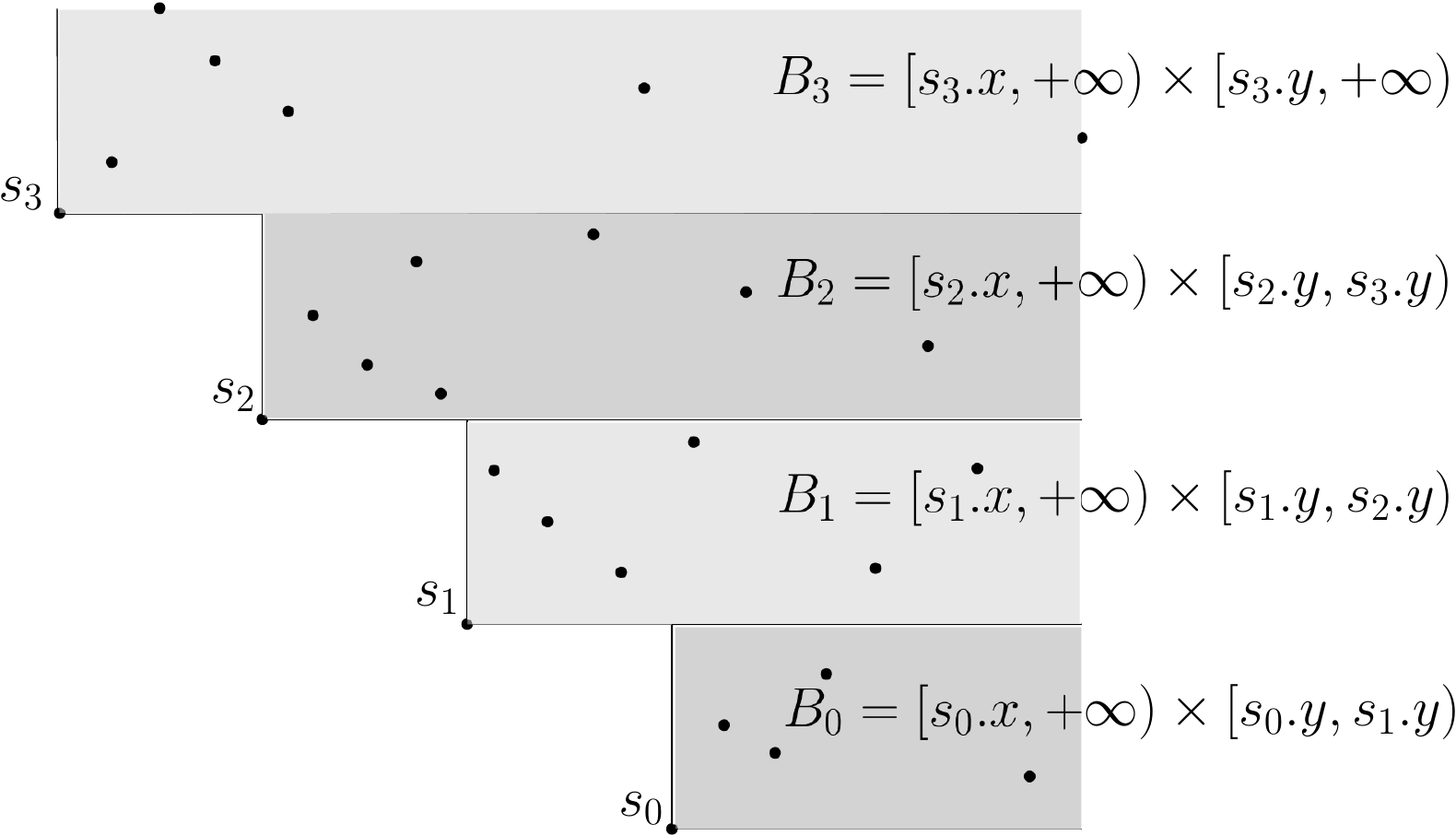}}
		\caption{\label{fig-box}
			An example of partitioning $U(A)$, where $A$ is a set of points on the plane, into $O(|A|)$ canonical boxes. In the figure, $s_0, s_1$, $s_2$, $s_3$ and $s_4$ are four points in $P$, and $B_0$, $B_1$, $B_2$ and $B_3$ are the canonical boxes $U(A)$ is partitioned into. All other points are in the interior of these boxes.}
	\end{figure}
	
	
	With this lemma, the reduction works as follows.
	Let $P$ denote the input colored point set of the 3D dominance colored range counting problem, and let $C$ denote the number of colors.
	Then, for each color $1\le c \le C$, we apply Lemma~\ref{lemma_box} to partition $U(P_c)$, where $P_c$ is the set of all the points in $P$ that are colored $c$, into a set, $B_c$, of $O(|U(P_c)|)$ disjoint 3D canonical boxes.
	We then construct a 3D stabbing counting structure over $B=\cup_{c=1}^C B_c$.
	Note that this data structure is constructed over $O(|P|)$ canonical boxes, as $\sum_{c=1}^C |B_c|=\sum_{c=1}^C O(|P_c|)=O(|P|)$.
	To answer a 3D dominance colored range counting query over $P$, for which the query range is the region dominated by a point $q$, observe that, if $q$ dominates at least one point in $P_c$, then it must be located with $U(P_c)$.
	Then, since $U(P_c)$ is partitioned into the boxes in $B_c$, we conclude that $q$ dominates at least one point in $P_c$ iff $q$ is contained in a box in $B_c$.
	Furthermore, since the boxes in $B_c$ are pairwise disjoint, $q$ is either contained in exactly one box in $B_c$ or outside $U(P_c)$.
	Hence, the number of distinct colors in the region dominated by $q$ is equal to the number of boxes in $B$ that contains $q$,
	which can be computed by performing a stabbing counting query in $B$ using $q$ as the query point. 
	
	\section{A New Framework of Achieving Time-Space Tradeoffs}
	\label{sect: almost_linear}
        We present three new solutions to 2D orthogonal colored range counting in this section and Section~\ref{sect:two_more}.
	They follow the same framework, of which we give an overview in Section~\ref{sec:framework}.
	One key component of this framework is a novel scheme of computing the sizes of the intersections between the color sets assigned to different subsets of points that lie within the query range; 
	Section~\ref{sect: new_tech} describes this scheme and shows how to combine it with Lemma~\ref{lemma:stabbing_4} to immediately achieve a new time-space tradeoff for 2D orthogonal colored range counting.
	%
	
	\subsection{Overview of the Data Structure Framework}
	\label{sec:framework}
	Let $P$ denote a set of $n$ points on the plane, each assigned a color identified by an integer in $[1, C]$.
	To support orthogonal colored range counting over $P$, we construct a binary range tree $T$ over the $y$-coordinates of the points in $P$ such that each leaf of $T$ stores a point of $P$, and, from left to right, the points stored in the leaves are increasingly sorted by $y$-coordinate.
	For each internal node $v$ of $T$, we construct the following data structures:
	\begin{itemize}
		\item A list $P(v)$ containing the points stored at the leaf descendants of $v$, sorted by $x$-coordinate; 
		\item A list $P_y(v)$ containing the sorted list of $y$-coordinates of the points in $P(v)$;
		\item 2D 3-sided colored range counting structures, $S(v_l)$ and, $S(v_r)$, constructed over $P(v_l)$ and $P(v_r)$ respectively, where $v_l$ and $v_r$ are respectively the left and right children of $v$ ($S(v_l)$ requires query ranges to open at the top while $S(v_r)$ requires them to open at the bottom; the exact data structures used will be selected later for different tradeoffs);
		\item 2D orthogonal colored range emptiness query structures, $E(v_l)$, and, $E(v_r)$, constructed over $\hat{P}(v_l)$ and $\hat{P}(v_r)$, respectively (using Lemma~\ref{lemma_range_emptiness}), where $\hat{P}(v_l)$ and $\hat{P}(v_r)$ are the point set in rank space converted from ${P}(v_l)$ and ${P}(v_r)$.
	\end{itemize}
	
	Let $Q=[a,b]\times[c,d]$ be the query rectangle. Given an internal node $v$, let $C_Q(v)$ denote the set of distinct colors assigned to points in $P(v) \cap Q$. 
	The query algorithm first locates the lowest common ancestor $u$ of the $c$-th and $d$-th leaves of $T$.
	As all the points from $P$ that are in the query range $Q$ must be in $P(u)$, $|C_Q(u)|$ is the answer to the query. 
	To compute  $|C_Q(u)|$, let $u_l$ and $u_r$ denote the left and right children of $u$, respectively.
	By the exclusion-inclusion principle, we know that $|C_Q(u)|= |C_Q(u_l)|+|C_Q(u_r)|- |C_Q(u_l)\cap C_Q(u_r)|$.
	Among the terms on the right hand side of this equation, $|C_Q(u_l)|$ and $|C_Q(u_r)|$ can be computed by performing 2D 3-sided colored range counting queries over $S(v_l)$ and $S(v_r)$, using $[a,b]\times[c,+\infty)$ and $[a,b]\times(-\infty, d]$ as query ranges, respectively. 
	What remains is to compute $|C_Q(u_l)\cap C_Q(u_r)|$.
	
	This idea of decomposing a 4-sided query range into two 3-sided query ranges has been used before for both 2D orthogonal colored range reporting \cite{gupta2018computational} and counting~\cite{kaplan2008efficient}.
	Furthermore, to support 2D 3-sided colored range counting, we apply the reduction of Kaplan et al.~\cite{kaplan2008efficient} (summarized in Section~\ref{sec:prelim} of our paper) to reduce it to the 3D stabbing query problems over canonical boxes.
	Thus, the techniques summarized so far have been used in previous work (without the construction of range emptiness structures).
	What is new is our scheme of achieving time-space tradeoffs when computing $|C_Q(u_l)\cap C_Q(u_r)|$;
	it gives us more {\em flexibility} in the design of 3D stabbing query structures, thus allowing us to achieve new results. 
	
	Here we describe the conditions that a 3D stabbing query structure must meet so that we can combine it with our scheme of computing $|C_Q(u_l)\cap C_Q(u_r)|$, while deferring the details of the latter to Section~\ref{sect: new_tech}.
	The stabbing query structure consists of multiple layers of trees of some kind. 
	The top-layer tree is constructed over the entire set of canonical boxes, and each of its nodes is assigned a subset of boxes.
	The second layer consists of a set of trees, each constructed over the boxes assigned to a node in the top-layer tree, and so on. 
	Thus each bottom-layer tree node is also assigned a list of boxes in a certain order (e.g., the $sList$ in the data structure for Lemma~\ref{lemma:stabbing_4}).
	The query algorithm locates a set $S$ of bottom-layer tree nodes. 
	For each node $v \in S$, there exists a nonempty prefix of the box list assigned to $v$, such that such prefixes over all the nodes in $S$ form a partition of the set of boxes containing the query point $q$.
	Furthermore, the size of each such prefix can be computed efficiently, and each box in such a prefix can also be reported efficiently.
	
	Clearly Lemma~\ref{lemma:stabbing_4} satisfies these conditions and can be used in our framework.
	On the contrary, even though Kaplan et al. proved this lemma and used it successfully in their $O(n^2\lg^2 n)$ space solution, they can not directly use it with their scheme of achieving time-space tradeoffs.
	Instead, they expand this structure with a third layer which is a segment tree constructed over $x$-coordinates of the boxes, increasing both time and space costs.
	Recall that in their scheme the set of boxes containing the query point are also decomposed into a fixed number of subsets.
	The reason Kaplan et al. cannot directly apply Lemma~\ref{lemma:stabbing_4} is that their scheme requires each of the  
	decomposed subsets to be equal to the entire set of boxes stored in a bottom-layer tree node, while our scheme allows each subset to be part of such a box set. 
	Hence, this extra flexibility allows us to use Lemma~\ref{lemma:stabbing_4} and alternative 3D stabbing query structures to be designed later (in Section \ref{sect:two_more}) in our framework.

	\subsection{Computing Intersections between Color Sets}
	\label{sect: new_tech}
	
	We now introduce our new scheme of computing $|C_Q(u_l)\cap C_Q(u_r)|$, and combine it with the stabbing query structure from Lemma~\ref{lemma:stabbing_4} to achieve a new time-space tradeoff for 2D orthogonal colored range counting.
	Since our scheme works with some other stabbing query structures, we describe it assuming a stabbing query structure satisfying the conditions described in Section~\ref{sec:framework} is used.
	To understand this scheme more easily, it may be advisable for readers to think about how it applies to the stabbing query structure of Lemma~\ref{lemma:stabbing_4}.
	
	
	Recall that, the problems of computing $|C_Q(u_l)|$ and $|C_Q(u_r)|$ have each been reduced to a 3D stabbing query.
	Furthermore, for each stabbing query, all reported boxes are distributed into a number of disjoint sets; the boxes in each such set form a nonempty prefix of the box list assigned to a node of the bottom-layer tree (henceforth we call each such box list a {\em bottom list} for convenience).
	Then, for the stabbing query performed to compute $|C_Q(u_l)|$, we define $D_Q$ to be a set in which each element is a such a disjoint set, and all these disjoint sets (whose union form the set of reported boxes) are elements of $D_Q$. 
	$U_Q$ is defined in a similar way for $|C_Q(u_r)|$. 
	Thus, if we use the data structure for Lemma~\ref{lemma:stabbing_4} to answer these stabbing queries, both $|D_Q|$ and $|U_Q|$ will be upper bounded by $O(\lg^2 n)$.
	As shown in Section~\ref{sec:prelim}, when reducing 2D 3-sided colored counting to 3D stabbing queries over canonical boxes, we guarantee that, each canonical box is part of the region $U(P_c)$ for each color $c$;
	we call this color the color of this canonical box, and explicitly store with each box its color. 
	Furthermore, for each color $1\leq c \leq C$, at most one canonical box colored in $c$ contains the query point, which implies that each box in $\cup_{s\in D_Q}s$ (resp. $\cup_{t\in U_Q}t$) has a distinct color.
	For each set $s\in D_Q$ and $t\in U_Q$, let $C(s)$ and $C(t)$ denote the set of colors associated with the boxes in $s$ and $t$, respectively.
	Then we have $|C_Q(u_l)| = \sum_{s\in D_Q} |C(s)|$, $|C_Q(u_r)|=\sum_{t\in U_Q} |C(t)|$, and $|C_Q(u_l)\cap C_Q(u_r)|=\sum_{s\in D_Q, t\in U_Q} |C(s)\cap C(t)|$.

	
	It now remains to show how to compute $\sum_{s\in D_Q, t\in U_Q} |C(s)\cap C(t)|$, for which more preprocessing is required.
	For each node $v$ in the binary range tree $T$, we construct a matrix $M(v)$ as follows: 
	Let $X \in [1,n]$ be a parameter to be chosen later.
	If the length, $m$, of a bottom list in the stabbing query structure $S(v_l)$ or $S(v_r)$ is greater than $X$, we divide the list into $\lceil m/X \rceil$ blocks, such that each block, with the possible exception of the last block, is of length $X$.
	If $m \le X$, then the entire list is a single block.
	If a block is of length $X$, we call it a {\em full block}. 
	Let $b_l(v)$ and $b_r(v)$ denote the total numbers of full blocks over all bottom lists in $S(v_l)$ and $S(v_r)$, respectively.
	Then $M(v)$ is a $b_l(v) \times b_r(v)$ matrix, in which each row (or column) corresponds to a nonempty prefix of a bottom list in $S(v_l)$ (or $S(v_r)$) that ends with the last entry of a full block,  
	and each entry $M[i,j]$ stores the number of colors that exist in both the set of colors assigned to the boxes in the prefix corresponding to row $i$ and the set of colors assigned to the boxes in the prefix corresponding to column $j$. 
	To bound the size of $M(v)$, we define the {\em duplication factor}, $\delta(n)$, of a stabbing query structure that satisfies the conditions in Section~\ref{sec:framework} to be the maximum number of bottom lists that any canonical box can be contained in. 
	For example, the duplication factor of the structure for Lemma~\ref{lemma:stabbing_4} is $O(\lg^2 n)$.
	(To see this, observe that in a segment tree constructed over $n$ intervals, each interval may be stored in $O(\lg n)$ tree nodes.
	Since segment trees are used in both layers of the data structure for Lemma~\ref{lemma:stabbing_4}, 
	each box can be stored in $O(\lg^2 n)$ different bottom lists.) 
	Since each full block contains $X$ boxes, the numbers of full blocks in the bottom lists of $S(v_l)$ and $S(v_r)$ are at most $\delta(n) |P(v_l)|/X$ and $\delta(n) |P(v_r)|/X$, respectively.
	Therefore, $M(v)$ occupies  $O(\delta(n)^2 |P(v_l)||P(v_r)|/X^2) = O((\delta(n) |P(v)|/X)^2)$ words. 
	
	With these matrices, the computation of $\sum_{s\in D_Q, t\in U_Q} |C(s)\cap C(t)|$ can proceed as follows. 
	Since each set $s \in D_Q$ (resp. $t\in U_Q$) occupies a prefix of a bottom list, $s$ (resp. $t$) can be split into two parts: $s_h$ (resp. $t_h$) which is the (possibly empty) prefix of $s$ (resp. $t$) that consists of all the full blocks  entirely contained in $s$ (resp. $t$), and $s_l$ (resp. $t_l$) which contains the remaining entries of $s$ (resp. $t$). 
	Thus we have $C(s_h) \cup C(s_l) = C(s)$ and $C(t_h)\cup C(t_l) = C(t)$.
	Since no two boxes in $s$ have the same color and the same applies to the boxes in $t$, $C(s_h) \cap C(s_l) = C(t_h)\cap C(t_l) = \emptyset$ also holds. Thus, we have

	\begin{equation} \label{eq_s_t_intersect}
	\begin{split}
	&\sum_{s\in D_Q, t\in U_Q} |C(s)\cap C(t)|  = \sum_{s\in D_Q, t\in U_Q} (|C(s_h)\cap C(t_h)|+|C(s_h)\cap C(t_l)| + |C(s_l)\cap C(t)|)\\
	& = \sum_{s\in D_Q, t\in U_Q} |C(s_h)\cap C(t_h)| +   \sum_{s\in D_Q, t\in U_Q} |C(s_h)\cap C(t_l)|  + |(\cup_{s\in D_Q} C(s_l))\cap \cup_{t\in U_Q} C(t)|
	\end{split}
	\end{equation}
	For the first term in the last line of Equation \ref{eq_s_t_intersect}, we can retrieve $|C(s_h)\cap C(t_h)|$ from the matrix $M(v)$ for each possible pair of $s_h$ and $t_h$ and sum them up. 
	Therefore, the first term can be computed in $O(|D_Q|\cdot |U_Q|)$ time. 
	For the third term, observe that, $\cup_{t\in U_Q} C(t) = C_Q(u_r)$. 
	Therefore, the third term can be computed by performing, for each color $c \in \cup_{s\in D_Q} C(s_l)$, a 2D orthogonal colored range emptiness query over $P(u_r)$ with $c$ as the query color and $Q$ as the query range.
	Note that the range emptiness query data structure $E(v_r)$ defined in Section~\ref{sec:framework} is built upon the points $\hat{P}(v_r)$ in rank space.
	We need to reduce $Q$ into rank space with respect to $\hat{P}(v_r)$ before performing colored range emptiness queries.
	Since all these queries share the same query range, we need only convert $Q$ into rank space once.
	This can be done by performing binary searches in $P(v)$ and $P_y(v)$ in $O(\lg n)$ time.
	As $|\cup_{s\in D_Q} C(s_l)| = |\cup_{s\in D_Q} s_l| = O(|D_Q|\cdot X)$, it requires $O(X |D_Q|\cdot (g(n)+\tau(n))+\lg n)$ time to compute these colors and then answer all these queries, where $g(n)$ denotes the query time of each range emptiness query in Lemma~\ref{lemma_range_emptiness} and $\tau(n)$ denotes the query time of reporting a box and its color in the query range.

	Finally, to compute the second term in the last line of Equation \ref{eq_s_t_intersect}, observe that,
	\begin{equation} \label{eq:h_l}
	\sum_{s\in D_Q, t\in U_Q} |C(s_h)\cap C(t_l)| = |(\cup_{s\in D_Q} C(s))\cap (\cup_{t\in U_Q} C(t_l))| - |(\cup_{s\in D_Q} C(s_l)) \cap (\cup_{t\in U_Q} C(t_l))|
	\end{equation}
	The first term of the right hand side of Equation \ref{eq:h_l} can be computed in $O(X |U_Q|\cdot (g(n)+\tau(n))+\lg n)$ time, again by performing range emptiness queries, but this time we use $E(v_l)$. 
	The second term can be computed by retrieving and sorting the colors in $\cup_{s\in D_Q} C(s_l)$ and those in $\cup_{t\in U_Q} C(t_l)$, and then scanning both sorted lists to compute their intersection.
	Since $|(\cup_{s\in D_Q} C(s_l))|$ (resp. $|(\cup_{t\in U_Q} C(t_l))|$) are bounded by $O(X|D_Q|)$ (resp. $O(X|U_Q|)$),
	the two sets of colors can be retrieved in $O(X|D_Q|\tau(n))$ and $O(X|U_Q|\tau(n))$ time and then sorted using Han's sorting algorithm \cite{han2002deterministic} in $O(X|D_Q|\lg\lg n)$ and $O(X|U_Q|\lg\lg n)$ time.
	Thus the second term in the last line of  Equation \ref{eq_s_t_intersect} can be computed in  $O(X(|U_Q|+|D_Q|)(\lg \lg n+\tau(n)))$ time.
	Overall, computing $|C_Q(u_l)\cap C_Q(u_r)|$ requires $O(|D_Q|\cdot |U_Q|+X(|U_Q|+|D_Q|)(\lg \lg n + g(n)+\tau(n))+\lg n)$ time. 
	Lemma~\ref{lemma:new_technique} summarizes the complexities of our framework. 
	
	\begin{lemma}
		\label{lemma:new_technique}
		Suppose that the 3D stabbing query structure of $S(v_l)$ (or $S(v_r)$) for each node $v \in T$ has duplication factor $\delta(n)$, occupies $O(|P(v)| h(n))$ words, and, given a query point $q$, it can compute  $\phi(n)$ disjoint sets of boxes whose union is the set of boxes containing $q$ in $O(\phi(n))$ time. 
		Furthermore, each subset is a nonempty prefix of a bottom list, and after this prefix is located, its length can be computed in $O(1)$ time and each box in it can be reported in $O(\tau(n))$ time.
		Let $f(n)$ and $g(n)$ be the functions set in Lemma~\ref{lemma_range_emptiness} to implement $E(v_l)$ and $E(v_r)$.
		Then the structures in our framework occupy $O((n\delta(n) /X)^2+n\lg n(f(n)+h(n)))$ words and answer a 2D orthogonal colored range counting query in $O(\phi^2(n) + X\phi(n)(\lg \lg n + g(n)+\tau(n)) +\lg n)$ time, where $X$ is an integer parameter in $[1, n]$.
	\end{lemma}

\begin{proof}
	Each node $v$ of $T$ stores a pair of lists, $P(v)$ and $P_y(v)$, of points, the colored range emptiness query structures $E(v_l)$ and $E(v_r)$, the stabbing query data structures $S(v_l)$ and $S(v_r)$, and the matrix $M(v)$.
	Among them, both $P(v)$ and $P_y(v)$ use $|P(v)|$ words of space. $E(v_l)$ and $E(v_r)$ use $O(|P(v)|f(n))$ words of space by  Lemma~\ref{lemma_range_emptiness}. $S(v_l)$ and $S(v_r)$ use $O(|P(v)|h(n))$ words of space.
	The matrix $M(v)$ uses $O((\delta(n) |P(v)|/X)^2)$ words of space. 
	Summing the space costs over all internal nodes of the range tree $T$, the overall space cost is at most $\sum_{v\in T} O((\delta(n) |P(v)|/X)^2 + |P(v)|(f(n)+h(n)))$.
	To simplify this expression, first observe that $\sum_{v\in T}|P(v)| = O(n\lg n)$.
	Furthermore, we can calculate $\sum_{v\in T}|P(v)|^2$ as follows:
	At the $i$-th level of $T$, there are $2^i$ nodes, and each node stores a point list of length $n/2^i$.
	Therefore, the sum of the squares of the lengths of the point lists at the $i$th level is $n^2/2^i$.
	Summing up over all levels, we have $\sum_{v\in T} |P(v)|^2 = O(n^2)$. 
	Therefore, the overall space cost simplifies to $O((n\delta(n) /X)^2+n\lg n(f(n)+h(n)))$. 
	
	Given a query range, we use $O(\phi(n))$ time to compute $|C_Q(u_l)|$ and $|C_Q(u_r)|$. 
	As shown before, computing $|C_Q(u_l)\cap C_Q(u_r)|$ can be reduced to computing $\sum_{s\in D_Q, t\in U_Q} |C(s)\cap C(t)|$, which requires $O(|D_Q|\cdot |U_Q|+X(|U_Q|+|D_Q|)(\lg \lg n + g(n)+\tau(n)) + \lg n)$ time.
	Since both $|D_Q|$ and $|U_Q|$ are upper bounded by $\phi(n)$, the overall query time is $O(\phi^2(n) + X\phi(n)(\lg \lg n + g(n)+\tau(n))+\lg n)$. 
\end{proof}
	
	We can now achieve a new time-space tradeoff by using the stabbing queries data structure from Lemma~\ref{lemma:stabbing_4} in our framework.
	To combine Lemmas~\ref{lemma:stabbing_4} and \ref{lemma:new_technique}, observe that $h(n) = O(\lg^2 n)$, $\phi(n) = O(\lg^2 n)$ and $\tau(n) = O(1)$.
	As discussed before, $\delta(n) = O(\lg^2 n)$.
	We use part b) of Lemma~\ref{lemma_range_emptiness} to implement $E(v_l)$ and $E(v_r)$, so $f(n) = O(\lg\lg n)$ and $g(n) = O(\lg\lg n)$.
	Hence: 
	
	\begin{theorem}
		\label{theorem: color_counting_0}
		Given $n$ colored points on the plane, there is a data structure of $O((\frac{n}{X})^2 \lg^4 n+ n\lg^3 n)$ words of space that answers colored orthogonal range counting queries in $O(\lg^4 n+X\lg^{2} n \lg \lg n)$ time, where $X$ is an integer parameter in $[1, n]$.
		In particular, setting $X=\sqrt{n\lg n}$ yields an $O(n\lg^3 n)$-word structure with $O(\sqrt{n}\lg^{5/2} n \cdot \lg \lg n)$ query time.
	\end{theorem}

	For the preprocessing, the binary range tree $T$ with all its auxiliary data structures except the matrices $M(v)$'s can be constructed in $O(n\lg^3 n)$ time.
	To construct the matrix $M(v)$, we can adopt the approaching used by Kaplan et al.~\cite{kaplan2008efficient} which computes a matrix over the heavy sets defined in their solution using matrix multiplication, by treating each block in our solution as a heavy set.
	Even though a block in our solution is different from a heavy set in theirs, the preprocessing still works.
	Thus we obtain the preprocessing time achieved by Kaplan et al.~\cite{kaplan2008efficient}.
	For example, when $X=\tilde{O}(\sqrt{n})$, the data structures can be constructed in $\tilde{O}(\frac{n^{(\omega+1)/2}}{X^{(\omega-1)/2}})=\tilde{O}(n^{1.343})$ time, where notation $\tilde{O}$ leaves out polylogarithmic factors and $\omega<2.3727$ denotes the exponent of matrix multiplication \cite{williams2012multiplying}.
	The same preprocessing approach also applies to other tradeoffs under our framework. 
	
	Unlike our result in Theorem~\ref{theorem: color_counting_0}, the solution of Kaplan et al.~\cite{kaplan2008efficient} with $O((\frac{n}{X})^2 \lg^6 n+ n\lg^4 n)$ words of space and $O(X\lg^7 n)$ query time works under the pointer machine model. 
	Nevertheless, with some modifications, our solution can also be made to work under this same model.
	First, Lemma~\ref{lemma_range_emptiness} requires the word RAM model, we can replace it by the optimal solution to the 2D orthogonal range emptiness query problem by Chazelle~\cite{chazelle1986filtering} with $O(n\lg n/\lg \lg n)$ words of space and $O(\lg n)$ query time.
	Thus, $g(n) = O(\lg n)$, but the overall space cost of the data structure remains unchanged.
	Second, when computing $|(\cup_{s\in D_Q} C(s_l)) \cap (\cup_{t\in U_Q} C(t_l))|$, we cannot use Han's sorting algorithm \cite{han2002deterministic} which requires the word RAM.
	Instead, using mergesort, we can compute this value in $O(X(|U_Q|+|D_Q|)\cdot \lg n)$ time.
	Finally, to simulate a matrix $M(v)$, we can use lists indexed by binary search trees, so that we can retrieve each entry in $O(\lg n)$ time. 
	Thus, we achieve the following result:
	\begin{corollary}
		\label{theorem: color_counting_pm}
		Under the arithmetic pointer machine model, given $n$ colored points on the plane, there is a data structure of $O((\frac{n}{X})^2 \lg^4 n+ n\lg^3 n)$ words of space that answers colored orthogonal range counting queries in $O(\lg^5 n+X\lg^{3} n)$ time, where $X$ is an integer parameter in $[1, n]$. In particular, setting $X=\sqrt{n\lg n}$ yields an $O(n\lg^3 n)$-word structure with $O(\sqrt{n}\lg^{7/2} n)$ query time.
	\end{corollary}
	
	\section{Two More Solutions with Better Space Efficiency}
	\label{sect:two_more}
As the space cost in Theorem~\ref{theorem: color_counting_0} is at least $\Omega(n\lg^3 n)$, we now design two more solutions with potentially better space efficiency for 2D orthogonal colored range counting. 
	
\subsection{Achieving $O(n\lg^2 n)$ Space}
	\label{sect: new-method-1}
	We design an alternative solution for 3D stabbing queries over canonical boxes whose space cost is a logarithmic factor less of that in Lemma~\ref{lemma:stabbing_4} asymptotically, and it also satisfies the conditions described in Section~\ref{sec:framework} and can thus be applied in our framework.
	This leads to another time-space tradeoff for 2D orthogonal colored range counting, whose space cost can be as little as $O(n\lg^2 n)$ by choosing the right parameter value.
	
	This new 3D stabbing query solution requires us to design a data structure supporting 2D dominance counting and reporting, 
%
	by augmenting a binary range tree constructed over the $y$-coordinates of the input points.
        Each internal node $v$ of $T$ is conceptually associated with a list, $P(v)$, of points that are leaf descendants of $v$, sorted by $x$-coordinate, but we do not store $P(v)$ explicitly.
        Lemma \ref{lemma: two_dimension_range_searching} presents this data structure.
	Even though better solutions exist for these problems~\cite{jaja2004space,chan2011orthogonal}, Lemma \ref{lemma: two_dimension_range_searching} gives us additional range tree functionalities that are required for our next two solutions to colored range counting.
	\begin{lemma}
		\label{lemma: two_dimension_range_searching}
		Consider a binary range tree $T$ constructed over a set, $P$, of $n$ points on the plane as described above.
		$T$ can be augmented using $O(n)$ additional words such that, given a query range $Q$ which is the region dominated by a point $q$, a set, $S$, of $O(\lg n)$ nodes of $T$ can be located in $O(\lg n)$ time that satisfies the following conditions:
		For each node $v \in S$, there exists a nonempty prefix $L(v)$ of $P(v)$ such that the point set $P \cap Q$ can be partitioned into $|S|$ disjoint subsets, each consisting of the points in such a prefix. 
		Furthermore, the individual sizes of all these subsets can be computed in $O(\lg n)$ time in total, and each point in such a subset can be reported in $O(\lg^{\epsilon} n)$ additional time for any positive constant $\epsilon$. 
	\end{lemma}

\begin{proof}
	For simplicity, we assume that point coordinates are in rank space; if we duplicate $P$ in two sorted sequences, one ordered by $x$-coordinate and the other by $y$-coordinate, the time required to perform the conversion between original coordinates and coordinates in rank space does not affect the claimed query time.
	Then each leaf of the range tree $T$ represents an integer range $[p.y, p.y]$ if $p$ is stored at this leaf.
	The range represented by an internal node of $T$ is the union of the ranges represented by its children. 
	
	At each internal node $v$ of $T$, we store a bit vector, $\mathcal{B}(v)$, such that if the point $P(v)[i]$ is a leaf descendant of the left child of $v$, then $\mathcal{B}(v)[i]$ is set to 0; otherwise $\mathcal{B}(v)[i]$ is set to 1.
	We construct a data structure of $O(|\mathcal{B}(v)|)$ bits of space upon $\mathcal{B}(v)$ to support the computation of $\rankop(v, k)$, which is $\sum_{i\leq k} B(v)[i]$, for any $k$ in constant time \cite{clark1996efficient}.
	These bit vectors over all internal nodes $v$ of $T$ use $\sum_v O(|\mathcal{B}(v)|)=O(n\lg n)$ bits, which is $O(n)$ words of space.
	
	Given a query range $Q=[1, q.x] \times [1, q.y]$, we find the path, $\pi$, from the root node of $T$ to the $q.y$-th leaf. 
	For each node $u$ in $\pi$, if it is the right child of its parent, we add its left sibling, $v$, into a set $S'$.
	We also add the $q.y$-th leaf into $S'$.
	Then the ranges represented by the nodes in $S'$ form a partition of the query $y$-range $[1, q.y]$, so $P \cap Q \subseteq \cup_{v\in S'} P(v)$. 
	Furthermore, for each node $v \in S'$, we add it into $S$ if $|P(v)\cap Q| > 0$. 
	To compute $|P(v)\cap Q|$, observe that, since the $y$-coordinates of points in $P(v)$ are within the query $y$-range, and these points are increasingly sorted by their $x$-coordinates, $|P(v)\cap Q|$ is equal to the index, $i$, of the rightmost point of $P(v)$ whose $x$-coordinate is no more than $q.x$. Furthermore, if $|P(v)\cap Q| > 0$, then $L(v) = P(v)[1..i]$.
	Thus the following observation is crucial: Let $s$ and $t$ be two nodes of $T$, where $s$ is the parent of $t$, let $j$ be the number of points in $P(s)$ whose $x$-coordinates are no more than $q.x$.
	Then, if $t$ is the left child of $s$, the number of points in $P(t)$ whose $x$-coordinates are no more than $q.x$ is $|P(s)| - rank(s, j)$.
	Otherwise, it is $\rankop(s, j)$.
	We know that at the root node $r$, the number of points in $P(r)$ whose $x$-coordinates are no more than $q.x$ is simply $q.x$.
	Then, during the top down traversal of $\pi$, we can make use of this observation and perform up to two $\rankop$ operations at each level, so that for any node $u \in \pi \cup S'$, we can compute the index of the rightmost point of $P(u)$ whose $x$-coordinate is no more than $q.x$.
	This way we can compute $S$ and for each $v \in S$, compute $|L(u)|$. The total running time is $O(\lg n)$. 
	
	Now we show how to report the coordinates of each point in $L(v)$ for any $v \in S$. 
	If $P(v)$ were explicitly stored at node $v$, then each point in the query range could be reported in constant time. 
	However, to store the lists for all internal nodes, it would require $O(n\lg n)$ words, which is not affordable.
	Instead, we define the operator, $\point(v, i)$, which computes coordinates of $P(v)[i]$ for each internal node $v$ of $T$. 
	To compute $\point(v, i)$, we can augment $T$ with the {\em ball inheritance} data structure \cite{chan2011orthogonal}, which can use $O(n)$ additional words of space to support $\point$ in $O(\lg^{\epsilon} n)$ time.
	Therefore, each point in the query range can be reported in $O(\lg^{\epsilon} n)$ time without explicitly storing $P(v)$.
\end{proof}
	
	Next we present the new stabbing query data structure.
	This time we construct data structures consisting of three layers of trees to answer stabbing queries over a given set of $n$ canonical boxes in three-dimensional space, but a different tree structure is adopted in each layer. 
	At the top-layer we construct a segment tree over the $z$-coordinates of the boxes.
	More precisely, we project each box onto the $z$-axis to obtain an interval, and the segment tree is constructed over all these intervals.
	A box is assigned to a node in this tree if its corresponding interval on the $z$-axis is associated with this node.
	For each node $v$ in the top-layer segment tree, we further construct an interval tree over the projections of the boxes assigned to $v$ on the $y$-axis, and the interval trees constructed for all the nodes of the top-layer tree form the middle-layer structure.
	For each node $v'$ of an interval tree, we use the set $B(v')$, of boxes assigned to it to define the two point sets, $S_{lower}(v')$ and $S_{upper}(v')$, on the plane as follows:
	We project all the boxes in $B(v')$ onto the $xy$-plane and get a set of right-open rectangles.
	Then, $S_{lower}(v')$ is the set of the lower left vertices of these rectangles (henceforth called {\em lower points}), i.e., $\{(B.x_1, B.y_1) | B\in B(v')\}$, and $S_{upper}(v')$ is the set of the upper left vertices (henceforth called {\em upper points}), i.e., $\{(B.x_1, B.y_2) | B\in B(v')\}$.
	We then use Lemma~\ref{lemma: two_dimension_range_searching} to build a pair of binary range trees, $T_{lower}(v')$ and $T_{upper}(v')$, over $S_{lower}(v')$ and $S_{upper}(v')$, respectively. 
	The range trees constructed for all these interval tree nodes form the bottom-layer structure.
	
	Recall that, each node $v''$ of a binary range tree in the bottom layer is conceptually associated with a list, $P(v'')$, of lower or upper points, which are the points stored in the leaf descendants of $v''$, sorted by $x$-coordinate.
	Each point of $P(v'')$ represents a box.
	Since there is a one-to-one correspondence between a point in $P(v'')$ and the box it represents, we may abuse notation and use $P(v'')$ to refer to the list of boxes that these points represent when the context is clear.
	Hence $P(v'')$ is the bottom list when the data structure is used in our framework.
	Lemma~\ref{lemma:stabbing_5} summarizes the solution.

	
	\begin{lemma}
		\label{lemma:stabbing_5}
		Given a set of $n$ canonical boxes in three dimension, the data structure above occupies $O(n\lg n)$ words and answers stabbing counting queries in $O(\lg^3 n)$ time and stabbing reporting queries in $O(\lg^3 n+k\cdot \lg^{\epsilon} n)$ time, where $k$ denotes the number of boxes reported. 
		Furthermore, the set of reported boxes is the union of $O(\lg^3 n)$ different disjoint sets, each of which is a nonempty prefix of some bottom list in the data structure. 
	\end{lemma}

\begin{proof}
	The top-layer segment tree occupies $O(n\lg n)$ words, while both interval trees and binary ranges trees from Lemma~\ref{lemma: two_dimension_range_searching} are linear-space data structures.
	Therefore, the overall space cost is $O(n\lg n)$ words. 
	
	To show how to answer a query, let $q$ be the query point. 
	Our query algorithm first searches for $q.z$ in the top-layer segment tree.
	This locates $O(\lg n)$ nodes of the top-layer tree. 
	Each node $v$ located in this phase stores a list, $B(v)$, of boxes whose $z$-ranges contain $q.z$.
	It now suffices to show how to count and report the boxes in each list $B(v)$ whose projections on the $xy$-plane contain $(q.x, q.y)$.
	To do this, we use the interval tree in the middle layer that is constructed over $B(v)$.
	In an interval tree, each node $v'$ stores the median, denoted by $m(v')$, of the endpoints of the intervals associated with its descendants (including itself), and it also stores a set, $I(v')$, of the intervals that contain $m(v')$.
	In our case, each interval in $I(v')$ is the $y$-range of a canonical box in $B(v)$, its left endpoint corresponds to the $y$-coordinate of a point in $S_{lower}(v')$, and its right endpoint corresponds to the $y$-coordinate of a point in $S_{upper}(v')$.
	The next phase of our algorithm then starts from the root, $r$, of this interval tree.
	If $q.y \le m(r)$, then an interval in $I(v')$ contains $q.y$ iff its left endpoint is less than or equal to $q.y$.
	This means, among the points in $S_{lower}(v')$, those lie in $(-\infty, q.x]\times(-\infty, q.y]$ correspond to the boxes whose projections on the $xy$-plane contain $(q.x, q.y)$.
	Since the $z$-range of these boxes already contains $q.z$, they contain $q$ in the three-dimensional space. 
	Hence, by performing a dominance query over $T_{lower}(v')$ using $(-\infty, q.x]\times(-\infty, q.y]$ as the query range, we can compute these boxes.
	Then, since the nodes in the right subtree $r$ store intervals whose left endpoints are greater than $m(r)$ which is at least $q.y$, none of these intervals can possibly contain $q.y$. Thus, we descend to the left child of $r$ afterwards and repeat this process.
	If $q.y > m(r)$ instead, then we perform a dominance query over $T_{upper}(v')$ using  $[-\infty, q.x]\times(q.y, +\infty)$ as the query range to compute the boxes associated with $r$ that contain $q$, descend to the right child of $r$, and repeat.

	To analyze the running time, observe that this algorithm locates $O(\lg n)$ nodes in the top-layer segment tree, and for each of these nodes, it further locates $O(\lg n)$ nodes in the middle-layer interval trees.
	Hence, we perform a 2D dominance counting and reporting query using Lemma~\ref{lemma: two_dimension_range_searching} for each of these $O(\lg^2 n)$ interval tree nodes, and the proof completes. 
\end{proof}
	
	In an interval tree, each interval is stored in exactly one node, while in a segment tree or a binary range tree, each interval or point can be associated with $O(\lg n)$ nodes.
	Therefore, this data structure has duplication factor $\delta = O(\lg^2 n)$. 
	If we combine Lemmas~\ref{lemma:stabbing_4} and \ref{lemma:new_technique}, we also have $h(n) = O(\lg n)$, $\phi(n) = O(\lg^3 n)$ and $\tau(n) = O(\lg^{\epsilon} n)$.
	We again use part b) of Lemma~\ref{lemma_range_emptiness} to implement $E(v_l)$ and $E(v_r)$, so $f(n) = O(\lg\lg n)$ and $g(n) = O(\lg\lg n)$.
	Hence: 
	
	\begin{theorem}
		\label{theorem: color_counting_1}
		Given $n$ colored points on the plane, there is a data structure of $O((\frac{n}{X})^2 \lg^4 n+ n\lg^2 n)$ words of space that answers colored orthogonal range counting queries in $O(\lg^6 n + X\lg^{3+\epsilon} n)$ time, where $X$ is an integer parameter in $[1, n]$ and $\epsilon$ is an arbitrary positive constant.
		In particular, setting $X=\sqrt{n} \lg n$ yields an $O(n\lg^2 n)$-word structure with $O(\sqrt{n}\lg^{4+\epsilon} n)$ query time.
	\end{theorem}
	
	\subsection{Achieving $O(n\lg n)$ Space}
	\label{sect: new-method-2}
	
	We further improve the space cost of the data structure for 3D stabbing queries over canonical boxes. 
	The new solution also satisfies the conditions described in Section~\ref{sec:framework} and can thus be applied in our framework.
	This leads to our third time-space tradeoff for 2D orthogonal colored range counting, whose space cost can be as little as $O(n\lg n)$ by choosing the right parameter value.
	
	Our solution will use a space-efficient data structure for 3D orthogonal range searching.
	Using a range tree with node degree $\lambda \in [2, n]$, we can transform a 2D linear space data structure shown in Lemma~\ref{lemma: two_dimension_range_searching} into a three-dimensional data structure that uses $O(n\log_{\lambda} n)$ words.
	Therefore, we can solve 3D dominance range searching as shown in Lemma~\ref{lemma:range_counting_reporting}.
	\begin{lemma}
		\label{lemma:range_counting_reporting}
		Given $n$ points in three dimension, there is a data structure of $O(n\log_{\lambda} n)$ words of space that answers 3D dominance counting queries in $O(\lambda \lg n \cdot \log_{\lambda} n)$ time and 3D dominance reporting queries in $O(\lambda \lg n \cdot \log_{\lambda} n+k\lg^{\epsilon} n)$ time, where $k$ is the number of reported points and $\lambda$ is an integer parameter in $[2, n]$.
	\end{lemma}

\begin{proof}
	The data structure is a balanced $\lambda$-ary range tree, $T$, constructed upon $z$-coordinates of the points.
	Each leaf of T stores an input point, from left to right, and all leaves are increasingly sorted by the $z$-coordinates of the points.
	At each internal node $v$, we explicitly store a list, $P(v)$, of points that are leaf descendants of $v$.
	Given $P(v)$, we construct a 2D dominance range searching data structure upon the $x$- and $y$-coordinates of points following Lemma~\ref{lemma: two_dimension_range_searching}.
	Since $T$ has $\log_{\lambda} n$ tree levels and each level stores data structures of $O(n)$ words, the data structure uses $O(n\log_{\lambda} n)$ words of space.
	
	Given a query range $Q=(-\infty, q.x] \times (-\infty, q.y]\times(-\infty, q.z]$, we find the path, $\pi$, from the root node of $T$ to the $z'$-th leaf, where $z'$ is the predecessor of $q.z$ and can be found in $O(\lg n)$ time by binary searches upon the leaf points.
	For each node $u$ in $\pi$, we add all its left siblings, $v$, into a set $S'$.
	We also add the $z'$-th leaf into $S'$.
	Then the ranges represented by the nodes in $S'$ form a partition of the query $z$-range $(-\infty, q.z]$, so the set of all reported points is a subset of $\cup_{v\in S'} P(v)$. 
	For each $v \in S'$, observe that, since the $z$-coordinates of points in $P(v)$ are within the query $z$-range, we can find the points $P(v)\cap Q$ and compute $|P(v)\cap Q|$ by 2D dominance searching upon $P(v)$ following Lemma~\ref{lemma: two_dimension_range_searching}.
	As $\pi$ has $O(\log_{\lambda} n)$ nodes and each node $v$ of $\pi$ has $O(\lambda)$ siblings, a dominance counting (resp. reporting) query takes $O(\lambda \lg n \cdot \log_{\lambda} n)$ (resp. $O(\lambda \lg n \cdot \log_{\lambda} n+k\lg^{\epsilon} n)$) time. 
\end{proof}

Next, we design a new data structure for stabbing queries over 3D canonical boxes.
It again contains trees of three layers, with interval trees in the top- and middle- layers plus the data structures for 3D dominance range searching from Lemma~\ref{lemma:range_counting_reporting} in the bottom layer.
More precisely, the structure at the top layer is an interval tree, $T_1$, constructed over the $z$-coordinates of the boxes. 
That is, we project each box onto the $z$-axis to obtain an interval, and the interval tree is constructed over all these intervals. 
A box is assigned to a node in this tree if its corresponding interval on the $z$-axis is associated with this node. 
For each node $v$ in the top-layer interval tree, we further construct an interval tree, $T_2(v)$, over the projections of the boxes assigned to $v$ on the $y$-axis, and the interval trees constructed for all the nodes of the top-layer tree form the middle-layer structure.
For each node $v'$ of an interval tree in the middle-layer, we use the set $B(v')$, of boxes assigned to it to define the four point sets, $S_{ll}(v')$, $S_{lr}(v')$, $S_{ul}(v')$, and $S_{ur}(v')$ in 3D such that $S_{ll}(v') = \{(B.x_1, B.y_1, B.z_1) | B\in B(v')\}$, $S_{lr}(v')=\{(B.x_1, B.y_1, B.z_2) | B\in B(v')\}$, $S_{ul}(v')=\{(B.x_1, B.y_2, B.z_1) | B\in B(v')\}$, and $S_{ur}(v')=\{(B.x_1, B.y_2, B.z_2) | B\in B(v')\}$.
We then use Lemma~\ref{lemma:range_counting_reporting} to build a set of 3D dominance range searching structures, $T_{ll}(v')$, $T_{lr}(v')$, $T_{ul}(v')$ and $T_{ur}(v')$, over $S_{ll}(v')$, $S_{lr}(v')$, $S_{ul}(v')$, and $S_{ur}(v')$, respectively. 
See Figure \ref{fig-point} for an illustration.
The 3D dominance range searching structures constructed for the nodes of all interval trees in the middle-layer form the bottom-layer structure.
	\begin{figure}[h]
		\centering
		{\includegraphics[scale=1]{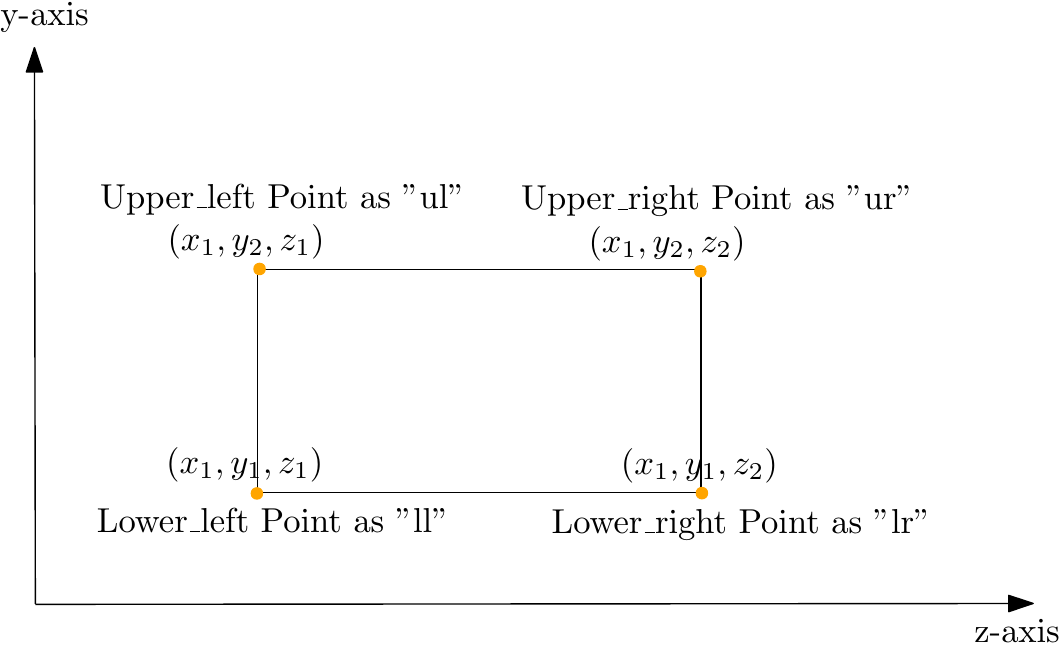}}
		\caption{\label{fig-point}
			The figure shows given a box projected on yz-plane, how we assign its endpoints into different sets $S_{ll}, S_{lr}, S_{ul}$, and $S_{ur}$.}
	\end{figure}
	
	As shown in the previous section, we need to identify the bottom lists from the data structures described above, which would be used in our framework. 
	Without loss of generality, we take the 3D dominance range searching structure, $T_{ll}(v')$, built in the bottom-layer as an example.
	Observe that $T_{ll}(v')$ is a $\lambda$-ary range tree, of which each internal node, $v''$, stores a binary range tree data structure $T_b(v'')$ implemented by Lemma \ref{lemma: two_dimension_range_searching} for 2D dominance range searching.
	Recall that, each node $\hat{v}$ of a binary range tree $T_b(v'')$ is conceptually associated with a list, $P(\hat{v})$, of points from the set $S_{ll}(v')$, which are the points stored in the leaf descendants of $\hat{v}$, sorted by $x$-coordinate.
	Each point of $P(\hat{v})$ represents a box.
	Since there is a one-to-one correspondence between a point in $P(\hat{v})$ and the box it represents, we may abuse notation and use $P(\hat{v})$ to refer to the list of boxes that these points represent when the context is clear.
	Hence $P(\hat{v})$ is the bottom list when the data structure is used in our framework.
	The following lemma summarizes this solution:
	\begin{lemma}
		\label{lemma:stabbing_general}
		Given a set of $n$ canonical boxes in three dimension, the data structure described above occupies $O(n\log_{\lambda} n)$ words of space and answers stabbing counting queries in $O(\lg^2 n \cdot \lambda \lg n \cdot \log_{\lambda} n)$ time and stabbing reporting queries in $O(\lg^2 n \cdot \lambda \lg n \cdot \log_{\lambda} n+ k \cdot \lg^{\epsilon} n)$ time, where $k$ denotes the number of boxes reported and $\lambda$ is an integer parameter in $[2, n]$.
		Furthermore, the set of reported boxes is the union of $O(\lg^2 n \cdot \lambda \lg n \cdot \log_{\lambda} n)$ different disjoint sets, each of which is a nonempty prefix of some bottom list in the data structure.
	\end{lemma}
	\begin{proof}
		The top- and middle- layer interval trees are linear-space data structure, while the 3D dominance range searching structures from Lemma~\ref{lemma:range_counting_reporting} occupy $O(n\log_{\lambda} n)$ words of space in total.
		Therefore, the overall space cost is $O(n\log_{\lambda} n)$ words. 
		
		To show how to answer a query, let $q=(q.x, q.y, q.z)$ be the query point.
		In an interval tree $T_1$ (resp. $T_2(v)$), each node $v$ (resp. $v'$) stores the median, denoted by $m_{z}(v)$ (resp. $m_{y}(v')$), of the endpoints of the intervals associated with its descendants (including itself), and it also stores a set, $I_{z}(v)$(resp. $I_{y}(v')$), of the intervals that contain $m_{z}(v)$ (resp. $m_{y}(v')$).
		In our case, each interval in $I_{z}(v)$ (resp. $I_{y}(v')$) is the $z$-range (resp. $y$-range ) of a canonical box in $B(v)$ (resp. $B(v')$).
		The query algorithm starts from the root, $r$, of $T_1$.
		If $q.z \leq m_z(r)$, then an interval in $I_z(r)$ contains $q.z$ iff its left endpoint is less than or equal to $q.z$.
		Then we visit the interval tree $T_2(r)$ in the middle-layer that is constructed upon the $y$-ranges of the boxes in $B(r)$.
		The next phase of our algorithm then starts from the root, $r'$, of $T_2(r)$.
		If $q.y \le m_y(r')$, then an interval in $I_y(r')$ contains $q.y$ iff its lower endpoint is less than or equal to $q.y$.
		This means, among the points in $S_{ll}(v')$, those lie in $(-\infty, q.x]\times(-\infty, q.y]\times(-\infty, q.z]$ correspond to the boxes  containing $(q.x, q.y, q.z)$.
		Hence, by performing a dominance query over $T_{ll}(v')$ using $(-\infty, q.x]\times(-\infty, q.y]\times(-\infty, q.z]$ as the query range, we can compute these boxes.
		Then, since the nodes in the right subtree $r'$ store intervals whose lower endpoints are greater than $m_y(r')$ which is at least $q.y$, none of these intervals can possibly contain $q.y$. 
		Thus, we descend to the left child of $r'$ afterwards and repeat this process.
		Otherwise if $q.y > m_y(r')$ instead, then we perform a dominance query over $T_{ul}(v')$ using  $(-\infty, q.x]\times(q.y, +\infty)\times(-\infty, q.z]$ as the query range to compute the boxes associated with $r'$ that contain $q$, descend to the right child of $r'$, and repeat this process until reaching the leaf level of $T_2(r)$. 
		Once $T_2(r)$ has been traversed, we return the root node $r$ of $T_1$.
		Since the nodes in the right subtree $r$ store intervals whose left endpoints are greater than $m_z(r)$ which is at least $q.z$, none of these intervals can possibly contain $q.z$. Thus, we descend to the left child of $r$ afterwards and repeat this process.
		Otherwise, if $q.z > m_z(r)$ instead, then we perform a dominance query over either $T_{lr}(v')$ using $(-\infty, q.x]\times(-\infty, q.y]\times(q.z, +\infty)$ as the query range or $T_{ur}(v')$ using $[-\infty, q.x]\times(q.y, +\infty)\times(q.y, +\infty)$ as the query range to compute the boxes associated with $v'$ that contain $q$, where $v'$ is a node of $T_2(r')$ that we visit, descend to the right child of $r$, and repeat the process.
		In summary, let $v$ denote a node on the traversed path of $T_1$ and given node $v$, let $v'$ denote a node on the traversed path of $T_2(v)$.
		By comparing $q.z$ against $m_{z}(v)$ and $q.y$ against $m_{y}(v')$, we can decide which 3D dominance range searching data structure to be used in the bottom-layer. 
		It includes the following four different cases:
		\begin{itemize}
			\item when $q.z\le m_{z}(v)$ and $q.y\le m_{y}(v')$, we search $T_{ll}(v')$ for the points of $S_{ll}(v')$ in the query range $(-\infty, q.x]\times(-\infty, q.y]\times(-\infty, q.z]$;
			\item when $q.z\le m_{z}(v)$ and $q.y > m_{y}(v')$, we search $T_{ul}(v')$ for the points of $S_{ul}(v')$ in the query range $(-\infty, q.x]\times(q.y, +\infty)\times(-\infty, q.z]$;
			\item when $q.z > m_{z}(v)$ and $q.y\le m_{y}(v')$, we search $T_{lr}(v')$ for the points of $S_{lr}(v')$ in the query range $(-\infty, q.x]\times(-\infty, q.y]\times(q.z, +\infty)$;
			\item when $q.z> m_{z}(v)$ and $q.y>m_{y}(v')$, we search $T_{ur}(v')$ for the points of $S_{ur}(v')$ in the query range $(-\infty, q.x]\times(q.y, +\infty)\times(q.z, +\infty)$.
		\end{itemize}

		To analyze the running time, observe that this algorithm locates $O(\lg n)$ nodes in the top-layer interval tree, and for each of these nodes, it further locates $O(\lg n)$ nodes in the middle-layer interval trees.
		Hence, we perform a 3D dominance counting and reporting query using Lemma~\ref{lemma:range_counting_reporting} for each of these $O(\lg^2 n)$ interval tree nodes, and the proof completes.
	\end{proof}

	In an interval tree, each interval is stored in exactly one node, while in a 3D dominance range searching structure from Lemma \ref{lemma:range_counting_reporting}, each point can be associated with $O(\lg n \cdot \log_{\lambda} n)$ nodes.
	Therefore, this data structure has duplication factor $\delta(n) = O(\lg n \cdot \log_{\lambda} n)$. 
	If we combine Lemmas~\ref{lemma:stabbing_general} and \ref{lemma:new_technique}, we have $h(n) = O(\log_{\lambda} n)$, $\phi(n) = O(\lg^2 n \cdot \lambda \lg n \cdot \log_{\lambda} n)$ and $\tau(n) = O(\lg^{\epsilon} n)$.
	We use part a) of Lemma~\ref{lemma_range_emptiness} to implement $E(v_l)$ and $E(v_r)$, so $f(n) = O(1)$ and $g(n) = O(\lg^{\epsilon} n)$.
	Hence:
	
	\begin{theorem}
		\label{theorem:ugly_tradeoff}
		Given $n$ colored points on the plane, there is a data structure of $O((\frac{n}{X})^2 \lg^2 n \cdot \log^2_{\lambda} n+n\lg n \cdot \log_{\lambda} n)$ words of space that answers colored orthogonal range counting queries in $O(\lambda^2\cdot \lg^6 n \cdot \log^2_{\lambda} n + X \cdot \lg^{3+\epsilon} n \cdot\lambda \log_{\lambda} n)$ time, where $X$ is an integer parameter in $[1, n]$, $\lambda$ is an integer parameter in $[2, n]$, and $\epsilon$ is any constant in $(0, 1)$.
		Setting $X=\sqrt{n \lg n \log_{\lambda} n}$ and $\lambda=\lg^{\epsilon} n$ yields an $O(n \frac{\lg^2 n}{\lg \lg n})$-word structure with $O(\sqrt{n}\lg^{5+\epsilon'} n)$ query time for any constant $\epsilon' > 2\epsilon$.
		Alternatively, setting $X=\sqrt{n \lg n}$ and $\lambda=n^{\epsilon/5}$ yields an $O(n\lg n)$-word structure with $O(n^{1/2+\epsilon})$ query time.
	\end{theorem}
	
	
	
	
\bibliography{main}

\begin{thebibliography}{10}

\bibitem{bw2012}
Nikhil Bansal and Ryan Williams.
\newblock Regularity lemmas and combinatorial algorithms.
\newblock {\em Theory Comput.}, 8(1):69--94, 2012.

\bibitem{c2015}
Timothy~M. Chan.
\newblock Speeding up the {F}our {R}ussians {A}lgorithm by {A}bout {O}ne {M}ore
  {L}ogarithmic {F}actor.
\newblock In {\em {SODA}}, pages 212--217, 2015.

\bibitem{chan2020further}
Timothy~M. Chan, Qizheng He, and Yakov Nekrich.
\newblock Further results on colored range searching.
\newblock In Sergio Cabello and Danny~Z. Chen, editors, {\em 36th International
  Symposium on Computational Geometry, SoCG 2020, June 23-26, 2020,
  Z{\"{u}}rich, Switzerland}, volume 164 of {\em LIPIcs}, pages 28:1--28:15.
  Schloss Dagstuhl - Leibniz-Zentrum f{\"{u}}r Informatik, 2020.

\bibitem{chan2011orthogonal}
Timothy~M Chan, Kasper~Green Larsen, and Mihai P{\u{a}}tra{\c{s}}cu.
\newblock Orthogonal range searching on the ram, revisited.
\newblock In {\em 27th Symposium on Computational Geometry}, pages 1--10. ACM,
  2011.

\bibitem{chan2020better}
Timothy~M Chan and Yakov Nekrich.
\newblock Better data structures for colored orthogonal range reporting.
\newblock In {\em Proceedings of the Fourteenth Annual ACM-SIAM Symposium on
  Discrete Algorithms}, pages 627--636. SIAM, 2020.

\bibitem{ccmn2000}
Moses Charikar, Surajit Chaudhuri, Rajeev Motwani, and Vivek~R. Narasayya.
\newblock Towards estimation error guarantees for distinct values.
\newblock In {\em {PODS}}, pages 268--279, 2000.

\bibitem{chazelle1986filtering}
Bernard Chazelle.
\newblock Filtering search: A new approach to query-answering.
\newblock {\em SIAM Journal on Computing}, 15(3):703--724, 1986.

\bibitem{chazelle1986fractional}
Bernard Chazelle and Leonidas~J Guibas.
\newblock Fractional cascading: I. a data structuring technique.
\newblock {\em Algorithmica}, 1(1-4):133--162, 1986.

\bibitem{clark1996efficient}
David~R Clark and J~Ian Munro.
\newblock Efficient suffix trees on secondary storage.
\newblock In {\em Proceedings of the seventh annual ACM-SIAM symposium on
  Discrete algorithms}, pages 383--391, 1996.

\bibitem{el2017succinct}
Hicham El{-}Zein, J.~Ian Munro, and Yakov Nekrich.
\newblock Succinct color searching in one dimension.
\newblock In Yoshio Okamoto and Takeshi Tokuyama, editors, {\em 28th
  International Symposium on Algorithms and Computation, {ISAAC} 2017, December
  9-12, 2017, Phuket, Thailand}, volume~92 of {\em LIPIcs}, pages 30:1--30:11.
  Schloss Dagstuhl - Leibniz-Zentrum f{\"{u}}r Informatik, 2017.

\bibitem{grossi2014colored}
Roberto Grossi and S{\o}ren Vind.
\newblock Colored range searching in linear space.
\newblock In {\em Scandinavian Workshop on Algorithm Theory}, pages 229--240.
  Springer, 2014.

\bibitem{gupta2018computational}
Prosenjit Gupta, Ravi Janardan, Saladi Rahul, and Michiel~HM Smid.
\newblock Computational geometry: Generalized (or colored) intersection
  searching.
\newblock {\em Handbook of Data Structures and Applications}, pages 1042--1057,
  2018.

\bibitem{gupta1995further}
Prosenjit Gupta, Ravi Janardan, and Michiel Smid.
\newblock Further results on generalized intersection searching problems:
  counting, reporting, and dynamization.
\newblock {\em Journal of Algorithms}, 19(2):282--317, 1995.

\bibitem{gupta1996algorithms}
Prosenjit Gupta, Ravi Janardan, and Michiel Smid.
\newblock Algorithms for generalized halfspace range searching and other
  intersection searching problems.
\newblock {\em Computational Geometry}, 6(1):1--19, 1996.

\bibitem{han2002deterministic}
Yijie Han.
\newblock Deterministic sorting in o (n log log n) time and linear space.
\newblock In {\em Proceedings of the thiry-fourth annual ACM symposium on
  Theory of computing}, pages 602--608, 2002.

\bibitem{pathcoloredcounting2021}
Meng He and Serikzhan Kazi.
\newblock Data structures for categorical path counting queries.
\newblock In {\em 32nd Annual Symposium on Combinatorial Pattern Matching,
  {CPM} 2021}, To appear.

\bibitem{jaja2004space}
Joseph J{\'a}J{\'a}, Christian~W Mortensen, and Qingmin Shi.
\newblock Space-efficient and fast algorithms for multidimensional dominance
  reporting and counting.
\newblock In {\em International Symposium on Algorithms and Computation}, pages
  558--568. Springer, 2004.

\bibitem{janardan1993generalized}
Ravi Janardan and Mario Lopez.
\newblock Generalized intersection searching problems.
\newblock {\em International Journal of Computational Geometry \&
  Applications}, 3(01):39--69, 1993.

\bibitem{kaplan2008efficient}
Haim Kaplan, Natan Rubin, Micha Sharir, and Elad Verbin.
\newblock Efficient colored orthogonal range counting.
\newblock {\em SIAM Journal on Computing}, 38(3):982--1011, 2008.

\bibitem{kaplan2006colored}
Haim Kaplan, Micha Sharir, and Elad Verbin.
\newblock Colored intersection searching via sparse rectangular matrix
  multiplication.
\newblock In {\em Proceedings of the twenty-second annual symposium on
  Computational geometry}, pages 52--60, 2006.

\bibitem{lai2008approximate}
Ying~Kit Lai, Chung~Keung Poon, and Benyun Shi.
\newblock Approximate colored range and point enclosure queries.
\newblock {\em Journal of Discrete Algorithms}, 6(3):420--432, 2008.

\bibitem{munro2015range}
J.~Ian Munro, Yakov Nekrich, and Sharma~V. Thankachan.
\newblock Range counting with distinct constraints.
\newblock In {\em Proceedings of the 27th Canadian Conference on Computational
  Geometry, {CCCG} 2015, Kingston, Ontario, Canada, August 10-12, 2015}, pages
  83--88. Queen's University, Ontario, Canada, 2015.

\bibitem{nekrich2014efficient}
Yakov Nekrich.
\newblock Efficient range searching for categorical and plain data.
\newblock {\em ACM Transactions on Database Systems (TODS)}, 39(1):1--21, 2014.

\bibitem{rahul2021approximate}
Saladi Rahul.
\newblock Approximate range counting revisited.
\newblock {\em Journal of Computational Geometry}, 12(1):40--69, 2021.

\bibitem{williams2012multiplying}
Virginia~Vassilevska Williams.
\newblock Multiplying matrices faster than coppersmith-winograd.
\newblock In {\em Proceedings of the forty-fourth annual ACM symposium on
  Theory of computing}, pages 887--898, 2012.

\bibitem{y2018}
Huacheng Yu.
\newblock An improved combinatorial algorithm for {B}oolean matrix
  multiplication.
\newblock {\em Inf. Comput.}, 261:240--247, 2018.

\end{thebibliography}

\end{document}